\documentclass[submission,copyright,creativecommons]{eptcs}
\providecommand{\event}{Express-SOS 2023} 

\usepackage{mathpartir}

\usepackage{graphicx}
\usepackage{bm,color}

\usepackage{comment}

\usepackage[conf={all end, no link to proof}]{proof-at-the-end}

\title{  Using $\pi$-Calculus Names as Locks
        }
\def\titlerunning{Using $\pi$-Calculus Names as Locks}

\author{Daniel Hirschkoff
  \institute{ENS de Lyon}
  \email{daniel.hirschkoff@ens-lyon.fr}
  \and
  Enguerrand Prebet
  \institute{Karlsruhe Institute of Technology}
  \email{enguerrand.prebet@kit.edu}
  }
  \def\authorrunning{D. Hirschkoff and E. Prebet}

\usepackage{enguemacros}

\definecolor{funcolor}{rgb}{.528,.266,.595}  

\newcommand{\api}{\ensuremath{A\pi}}
\newcommand{\piL}{\ensuremath{\pi \mathrm{\lel}}}

\newcommand{\piw}{\ensuremath{\pi \mathrm{\lel w}}}

\newcommand{\CCSL}{\ensuremath{\mathrm{CCS}\lel}}

\newcommand{\lamlo}{\ensuremath{\lambda_{\mathrm{lock}}}}
  
\newcommand{\tlockw}[2]{\ensuremath{\langle#1\rangle^{#2}}}
\renewcommand{\tlockw}[2]{\ensuremath{\langle#1\rangle_{#2}}}

\newcommand{\tlock}[1]{\ensuremath{\langle#1\rangle}}

\newcommand{\lel}{\ensuremath{\ell}}
\newcommand{\obs}{\ensuremath{o}}

\renewcommand{\typ}[2]{\ensuremath{#1\vdash#2}}

\newcommand{\Rs}{\ensuremath{\mathbb{R}}}

\newcommand{\G}{\ensuremath{\Gamma}}
\newcommand{\g}{\ensuremath{\gamma}}

\newcommand{\emps}{\ensuremath{\emptyset}}
\newcommand{\Gop}[2]{\ensuremath{#1\bullet #2}}
\newcommand{\Gopl}[3]{\ensuremath{#1\bullet_{#2} #3}}
\newcommand{\Gapp}[1]{\ensuremath{\mathsf{flatten} #1}}

\newcommand{\connname}{\ensuremath{\mathsf{connect}}}
\newcommand{\conn}[2]{\ensuremath{\connname(#1;#2)}}
\newcommand{\Pdl}{\ensuremath{P_{\mathsf{dl}}}}

\newcommand{\ttyp}{\ensuremath{\tau}}
\renewcommand{\ttyp}{\ensuremath{\mathsf{T}}}

\newcommand{\bbool}{\ensuremath{\mathsf{bool}}}
\newcommand{\ffalse}{\ensuremath{\mathsf{ff}}}
\newcommand{\ttrue}{\ensuremath{\mathsf{tt}}}
\newcommand{\bval}{\ensuremath{\mathtt{b}}}

\newcommand{\encw}[1]{\ensuremath{[\![#1]\!]_w}}
\newcommand{\lset}{\ensuremath{\mathbb S}}
\renewcommand{\nil}{\ensuremath{\bm{0}}}
\newcommand{\wait}[2]{\ensuremath{#1(\!(#2)\!)}}
\newcommand{\wtau}[1]{\ensuremath{\tau/#1}}

\newcommand{\tstra}[5]{\ensuremath{[{#1};{#2}]\stra{#3}[#4;{#5}]}}

\newcommand{\flocks}[1]{\ensuremath{\mathrm{fln}(#1)}}
\newcommand{\blocks}[1]{\ensuremath{\mathrm{bln}(#1)}}

\newcommand{\sharel}[1]{\ensuremath{\stackrel{#1}{\longleftrightarrow}}}
\newcommand{\nwbe}{\ensuremath{\not\wbe}}
\newcommand{\wbeo}{\ensuremath{\wbe_o}}
\newcommand{\wbew}{\ensuremath{\wbe_w}}
\newcommand{\bisL}{\ensuremath{\approx}}
\newcommand{\nbisL}{\ensuremath{\not\bisL}}
\newcommand{\bisw}{\ensuremath{\approx_w}}

\begin{document}
\longfalse

\newtheorem{definition}{Definition}
\newtheorem{example}[definition]{Example}
\newtheorem{lemma}[definition]{Lemma}
\newtheorem{remark}[definition]{Remark}
\newtheorem{proposition}[definition]{Proposition}
\newtheorem{theorem}[definition]{Theorem}
\newtheorem{corollary}[definition]{Corollary}

\maketitle
\begin{abstract}
  Locks are a classic data structure for concurrent programming. 
  We introduce a type system to ensure that names of the asynchronous
  $\pi$-calculus are used as locks.
    Our calculus also features a construct to deallocate a lock once
  we know that it will never be acquired again.
    Typability  guarantees two properties:
  deadlock-freedom, that is, no acquire operation on a lock waits forever; and
  leak-freedom, that is, all locks are eventually deallocated.

  We leverage the simplicity of our typing discipline to study the
  induced typed behavioural equivalence. After defining barbed
  equivalence, we introduce a sound labelled bisimulation, which
  makes it possible to establish equivalence between programs that
  manipulate and deallocate locks.
        \end{abstract}

\section{Introduction}\label{s:intro}

The $\pi$-calculus is an expressive process calculus based on the
notion of name, in which name-passing is the primitive notion of
interaction between processes. Processes of the $\pi$-calculus have been used to
represent several aspects of programming, like 
data structures, protocols, or  constructs such as functions, continuations, objects, and
references.
The $\pi$-calculus also comes with a well-developed theory
of behavioural equivalence. This theory can be exploited to reason about contextual equivalence in
programming languages, by translating programs as $\pi$-calculus processes.

In this work, we follow this path for locks, a basic data structure
for concurrent programming. We study how $\pi$-calculus names can be
used to represent locks. We show that the corresponding programming discipline
in the $\pi$-calculus induces a notion of behavioural
equivalence between processes, which can be used to reason about
processes manipulating locks.
This approach has been followed to analyse several disciplines for
the usage of $\pi$-calculus names:
linearity~\cite{DBLP:journals/toplas/KobayashiPT99},
receptiveness~\cite{DBLP:conf/icalp/Sangiorgi97},
locality~\cite{DBLP:journals/mscs/MerroS04},
internal mobility~\cite{DBLP:journals/tcs/Sangiorgi96a},
functions~\cite{DBLP:journals/iandc/Sangiorgi94,DBLP:conf/lics/DurierHS18}, references~\cite{DBLP:conf/concur/HirschkoffPS20,DBLP:conf/icalp/Prebet22}.

It is  natural to represent locks in \api, the asynchronous version of the
$\pi$-calculus~\cite{DBLP:journals/tcs/AmadioCS98,DBLP:conf/ecoop/HondaT91}. A lock is referred
to using a $\pi$-calculus name. It is represented as an asynchronous output: the
release of the lock. Dually, an input represents the acquire operation
on some lock.

\smallskip

In this paper, we introduce \piw, a version of the asynchronous
$\pi$-calculus with only lock names. 
Two properties should be ensured for names to be used as locks: first,
a lock can appear at most once in released form. Second, 
acquiring a lock entails the obligation to release it.
{For instance, process
  $\inp{\lel_1}x.(\out{\lel_1}x|\out{\lel_2}x)$ has these  properties: the process acquires lock $\lel_1$, then releases it,
  together with lock $\lel_2$.
    We remark that this this process owns lock $\lel_2$, which is released after $\lel_1$ is
  acquired.
We show that a simple type system can be defined to guarantee the
two  properties mentioned above.
}

\smallskip

When manipulating locks, it is essential to avoid the
program from getting stuck in a state where a lock needs to be
acquired but cannot be released.
Consider the following process:
$$\Pdl\quad\defi\quad \lel_1(x).(\out{\lel_1}x|\out{\lel_2}x)
~~|~~
\lel_2(y).(\out{\lel_1}y|\out{\lel_2}y)
.$$The subprocess on the left needs to acquire lock $\lel_1$, which is
owned by the other subprocess, and symmetrically: this is a deadlock.
Our type system rules out processes that exhibit this kind of cyclic
dependency between locks. 
This is achieved by
controlling parallel composition: two
processes in parallel can share at most one lock name.
Process \Pdl{} thus cannot be typed,
because names $\lel_1$ and $\lel_2$ are shared
between the two subprocesses. The  acyclicity property enjoyed by typable processes yields
deadlock-freedom.

\smallskip

To avoid situations where a lock is in released state and cannot be
accessed, \piw{} also features a construct to deallocate a lock, called  \emph{wait},  
inspired from~\cite{DBLP:journals/pacmpl/JacobsB23}. Process $\wait\lel{x}.P$ waits until
no acquire is pending on lock \lel, at which point it deallocates
\lel, reading the final value stored in \lel{} as $x$. The reduction rule for wait
is
\begin{equation}
  \label{eq:wait}
  (\new\lel)\big(\,\out\lel v~|~ \wait\lel{x}.P\,\big)
  ~\red~
  P\sub vx
    \end{equation}
\noindent
provided \lel{} is not among the free names of $P$.
In the reduction above, the restriction on
\lel{} disappears after the last interaction involving \lel{} has taken
place.

\medskip

The main contributions of this work are the following:
\begin{itemize}
\item We introduce \piw, a $\pi$-calculus with higher-order locks
  (in the sense that locks can be stored in locks)
  and   lock deallocation. The type system for \piw{} controls 
      the usage and the sharing of lock 
  names between processes.
  We provide some examples to illustrate how locks can be manipulated
  according to the programming discipline induced by types.

\item We show that typable processes in \piw{} enjoy deadlock- and leak-freedom.
  The proofs rely on simple arguments involving   the graph induced by the
  sharing of locks among processes.

\item We analyse {typed behavioural equivalence} in \piw.  Types
  restrict the set of contexts that can interact with processes,
  yielding a coarser behavioural equivalence than in the untyped case.

  We first introduce typed barbed equivalence, written
  \wbew. Relation \wbew{} is defined by observing the behaviour of processes
  when they are placed in typable contexts. 
    We then express the interactions between typed processes and typed
  context by means of a Labelled Transition System (LTS) that takes
  into account typing constraints. This allows us to introduce
  \emph{typed bisimilarity}, \bisw, the main proof technique to
  establish barbed equivalence: we indeed prove soundness, that
  is, $\bisw\subseteq\wbew$.

  We discuss several examples that help to understand how we can
  reason about behavioural equivalence in \piw.
          We are not aware of existing labelled equivalences taking into
  account name deallocation in the $\pi$-calculus.
    \end{itemize}
Beyond \piw, we believe that \bisw{} can be used
as a building block when reasoning in the $\pi$-calculus about
programs that use various features, among which locks.

  \medskip
  
The aforementioned contributions are presented in two steps. We first
introduce \piL, an asynchronous $\pi$-calculus with higher-order
locks. \piw{} is obtained by adding the wait construct to
\piL. Several important ideas can be presented in \piL, and we
can build on the notions introduced for \piL{} to extend them for
\piw.

\smallskip

We now highlight some of the technical aspects involved in our work.

The type system for \piL{} guarantees deadlock-freedom, in the sense
that for typable processes, an acquire operation cannot be blocked
forever. This holds for \emph{complete processes}: a process is
complete if for every lock \lel{} it uses, a release of \lel{} is
available. Availability need not be immediate, for instance the
release operation on lock \lel{} may be blocked by an acquire on $\lel'$. We
prove progress based on the fact that the
type system guarantees acyclicity of the dependence relation between
locks. Progress  entails deadlock-freedom.

When adding the wait construct, we rely on a similar reasoning to
prove leak-freedom for \piw, which in our setting means that all locks
are eventually deallocated.
The type system for \piw{} is richer than the one for \piL{} not only
because it takes wait into account, but also because it makes it
possible to transmit the obligation of releasing or deallocating a
lock via another lock. For instance, it is possible, depending on the
type of \lel, that in process $\lel(\lel').P$, the continuation $P$
has the obligation not only to release lock \lel, but also to
deallocate $\lel'$, or release $\lel'$, or both.

\smallskip

To define typed barbed equivalence in \piL, written \wbe, we must take
into account deadlock-freedom, which has several consequences.  First,
we observe complete processes: intuitively, computations in \piL{}
make sense only for such processes, and a context interacting with a process
should not be able to block a computation by never performing some
release operation.
Second, all barbs are always observable in \piL. In other words, if
\lel{} is a free name of a complete typable process $P$, then $P$ can
never loose the ability to release \lel. This is in contrast with
barbed equivalence in the $\pi$-calculus, or in CCS, where the absence
of a barb can be used to observe behaviours.
We therefore adopt a stronger notion of barb, where the value stored
in a lock, and not only the name of the lock, can be observed.

The ideas behind \wbe{} are used to define \wbew, typed barbed
equivalence in \piw.
A challenge when defining typed bisimilarity in \piw{} is to come
up with labelled transitions corresponding to the reduction
in~\rref{eq:wait}. Intuitively, if $P\stra{\wait\lel v}P'$ ($P$
deallocates \lel{} and continues as $P'$), we must make sure that this
transition is the last interaction at \lel.
We define a typed LTS to handle name deallocation, and show that
bisimilarity is sound for barbed equivalence in \piw.

\paragraph*{Paper outline.}
We study \piL{} in Section~\ref{s:piL}.  We
first expose the essential ideas of our
deadlock-freedom proof in \CCSL, a simple version of the Calculus of
Communicating Systems~\cite{DBLP:books/sp/Milner80} with lock
names. After extending these results to \piL, we define barbed equivalence
for \piL, written \wbe. We provide a labelled semantics that is sound
for \wbe, and present several examples of behavioural equivalences in
\piL.
In Section~\ref{s:piw}, we add the wait construct, yielding  \piw. We
show how to derive leak-freedom, and define a labelled
semantics, building on the ideas of Section~\ref{s:piL}.
We discuss related and future work in Section~\ref{s:ccl}.

\section{\piL, a Deadlock-Free Asynchronous $\pi$-Calculus}\label{s:piL}

We present deadlock-freedom in the simple setting of \CCSL{} in Section~\ref{s:CCSL}. This
approach is extended to handle higher-order locks in \piL{} (Section~\ref{s:piL:deadlock}). We
study behavioural equivalence in \piL{} in Section~\ref{s:piL:behav}.

\subsection{\CCSL: Ensuring Deadlock-Freedom using Composition}\label{s:CCSL}

\CCSL{} is a simplification of \piL, to present the ideas
underlying the type system and the proof of deadlock-freedom.
\CCSL{} is defined as an asynchronous version of CCS with acquire and
release operations.
We postulate the existence of an infinite set of \emph{lock names},
written $\lel, \lel',\lel_1,\dots$,
   which we often simply call names.
 \CCSL{} processes are defined by the following grammar:
\begin{mathpar}
  P~::=~ \ell.P\OR\outC\ell   \OR(\new\ell) P\OR P_1|P_2.
\end{mathpar}
\outC\lel{} is the release of lock \lel. Process $\lel.P$ acquires
\lel{} and then acts as $P$---we say that $P$ performs an \emph{acquire on \lel}.
There is no  \nil{} process in \CCSL, intuitively because we do not take into
consideration processes with no lock at all.
Restriction is a binder, and we write \flocks P for the set of free
lock names in $P$.
If $\lset=\set{\lel_1,\dots,\lel_k}$ is a set of lock names, we write $(\new\lset)P$ for
$(\new\lel_1)\dots(\new\lel_k)P$.

The definition of structural congruence, written $\equiv$, and
reduction, written \stra{}, are standard. They are given in
Appendix~\ref{a:CCSL}. Relation \wtra{} is the transitive reflexive
closure of \stra{}.

\paragraph*{Type System.}

To define the type system for \CCSL, we introduce typing environments.
We use \g{} to range over sets of lock names.
We write $\g_1\#\g_2$ whenever $\g_1\cap\g_2=\emptyset$.
We write $\g,\ell$ for the set
$\g\uplus\set{\ell}$: the notation implicitly imposes $\ell\notin\g$.

\emph{Typing environments}, written \G, are sets of such sets, with
the additional constraint that these should be pairwise disjoint. We
write $\G=\g_1,\dots,\g_k$, for $k\geq 1$, to mean that \G{} is equal
to $\set{\g_1,\dots,\g_k}$, with $\g_i\#\g_j$ whenever $i\neq
j$.
The
$\g_i$s are called the \emph{components} of \G{} in this case, and
\dom\G, the domain of \G, is defined as $\g_1\cup\dots\cup\g_k$.
We write
$\G_1\#\G_2$ whenever $\dom{\G_1}\cap\dom{\G_2}=\emptyset$.

As for components \g, the notation $\G,\g$ stands for a set (of sets) that
can be written as $\G\uplus\set\g$.
Using these two notations together, we can write $\G,\g,\ell$ to
refer to a typing environment containing a component that contains
\lel. 
We sometimes add parentheses, writing e.g.\  $\G,(\g,\ell,\lel')$, to ease
readability.

The typing judgement is of the form \typ{\G;\Rs}P, where \Rs{} is a set of lock
  names.
  If \typ{\G;\Rs}P, then \Rs{} is the set of
locks  owned by $P$, that
must be released. Moreover any component \g{} of \G{} intuitively corresponds to
a subprocess of $P$ that only accesses the names in \g. Here,
accessing a lock name \lel{} means either releasing
\lel{} or performing an acquire on \lel, or both.
The typing rules are as follows:
\begin{mathpar}
    {
    \inferrule[\trans{Acq-C}]{\typ{\G,(\g,\ell);\Rs,\ell}P }{
      \typ{{\set{\Gapp(\G)\uplus (\g,\ell)}};\Rs}{\ell.P} }
    }

  \inferrule[\trans{Rel-C}]{ }{
    \typ{\G,(\g,\ell);\set\ell}{\outC\ell}    }

  \inferrule[\trans{New-C}]{ \typ{\G,(\g,\ell);\Rs,\ell}P }{ \typ{\G,\g;\Rs}{(\new\ell) P}
 }

 \inferrule[\trans{Par-C}]{ \typ{\G_1;\Rs_1}{P_1}\and \typ{\G_2;\Rs_2}{P_2} }{
   \typ{\Gop{\G_1}{\G_2};\Rs_1\uplus \Rs_2}{P_1|P_2} }
\end{mathpar}
  In rule \trans{Acq-C}, operator \Gapp{} has the effect of mergining
all components in a typing environment into a single component. In particular, if
$\G=\set{\g_1,\dots,\g_k}$, then \Gapp(\G) stands for
$\g_1\uplus\dots\uplus\g_k$.
Intuitively, the causal dependency introduced by
the prefix $\lel.P$ induces a dependence between \lel{} and all the locks in $P$, forcing these  locks to belong to
the same component.

In the typing rules, we write $\Rs,\ell$ for $\Rs\uplus\set\ell$,
i.e., we suppose $\ell\notin \Rs$, otherwise
the typing rule cannot be applied.
Lock \lel{} is added to \Rs{}  in rule
\trans{Acq-C}, to ensure that it will be released in the continuation
$P$, and in rule \trans{New-C}, to ensure that a newly created lock is
initialised with a release.
Correspondingly, rule \trans{Rel-C} type-checks the release of lock
\lel{} by imposing $\Rs=\set\lel$.

To type-check parallel composition, we use an operation to {compose typing environments}, written
\Gop{\G_1}{\G_2}.  For this, we
set 
$\Gop\emptyset\G = \G$ and $\Gop{(\G,\g)}{\G'} = \conn\g{\Gop\G{\G'}}$,
where \conn\g{\set{\g_1,\dots,\g_k}} is undefined as soon as there is $i$
such that $\g\cap\g_i$ contains at least two distinct elements, and otherwise
is defined as
\begin{mathpar}
   \conn\g{\set{\g_1,\dots,\g_k}} \quad=\quad
   \set{\g_i: \g_i\#\g} ~\uplus~ \set{\g\cup \Gapp(\set{\g_i:
     \g_i\cap\g      \neq\emptyset})}
  .
\end{mathpar}
In rule \trans{Par-C}, we impose that $\Rs_1$ and $\Rs_2$ are
disjoint: if lock \lel{} must be released, then this is done either by $P_1$ or
by $P_2$. Together with rule \trans{Rel-C}, this guarantees that any
$\lel\in\Rs$ is released
 exactly once. 

We present some examples to illustrate the type system.
\begin{example}\label{e:Gop}
  Processes
        $\lel_1.(\outC{\lel_1}|\outC{\lel_1})$
        and
        $\lel_1.\lel_2.\outC{\lel_1}$ 
        cannot be typed,
        because both violate linearity in the usage of locks:
        the former  releases lock
$\lel_1$ twice, and
        the latter does
        not release $\lel_2$ after acquiring it.

Process $P_1\defi \lel_1.(\outC{\lel_1}|\outC{\lel_2})$ acquires lock
$\lel_1$, and then releases locks $\lel_1$ and $\lel_2$. Let
$\g_{12} = \set{\lel_1,\lel_2}$; we can derive
\typ{\set{\g_{12}};\set{\lel_2}}{P_1}: locks $\lel_1$ and
$\lel_2$ necessarily belong to the same component when typing $P_1$.
Similarly, we have \typ{\set{\g_{12}};\set{\lel_1}}{P_2} with $P_2
\defi  \lel_2.(\outC{\lel_2}|\outC{\lel_1})$.
The typing derivations for $P_1$ and $P_2$ cannot be composed, because
of the presence of $\g_{12}$ in both, so $P_1|P_2$ cannot be
typed. This is appropriate, since $P_1|P_2$ presents a typical deadlock
situation, where $\lel_1$ is needed to release $\lel_2$ and
conversely.

On the other hand, process
 $P_3 \defi \lel_1.(\outC{\lel_1}|\outC{\lel_2}) ~|~
 \lel_2.\outC{\lel_2} ~|~ \lel_1.\outC{\lel_1}$ can be typed: we can derive
 \typ{\set{\g_{12}};\set{\lel_2}}{\lel_1.(\outC{\lel_1}|\outC{\lel_2})} and
 \typ{\set{\set{\lel_1},\set{\lel_2}};\emps}{ \lel_2.\outC{\lel_2} ~|~
   \lel_1.\outC{\lel_1}}, and we can compose these typing derivations,
 yielding \typ{\set{\g_{12}};\set{\lel_2}}{P_3}. Crucially, 
 components $\set{\lel_1}$ and $\set{\lel_2}$ are not merged in the second derivation
 for the composition to be possible.
Using similar ideas, we can define a typable process made of three
parallel components $P_1, P_2, P_3$ sharing a single lock, say \lel,
as long as each of the $P_i$ uses its own locks besides \lel.

We can derive 
 \typ{\set{\g_{12}};\emps}{P_4} with $P_4
\defi  \lel_1.\lel_2.(\outC{\lel_2}|\outC{\lel_1})$. We observe that
$P_4|P_4$ cannot be typed, although $P_4|P_4$ is `no more deadlocked'
than $P_4$ alone.

\end{example}

The typing rules enforce $\Rs\subseteq\dom\G$ when
deriving \typ{\G;\Rs}P.
  We say that $\ell$ is \emph{available} in process $P$ if $P$
  contains   a release of $\ell$ which is not under
  an acquire on $\ell$ in $P$.
  Intuitively, when \typ{\G;\Rs}P is derivable, $P$ is a well-typed process in which
  all lock names in \Rs{} are available in $P$. The type system thus
  guarantees a linearity property on the release of names in
  \Rs. However, lock names are not \emph{linear names} in the sense
  of~\cite{DBLP:journals/toplas/KobayashiPT99}, since there can be
  arbitrarily many  acquire operations on a given lock. 
  {When all free lock names are available in $P$, i.e. $\typ{\G;\flocks{P}}P$,
  we say that $P$ is \emph{complete}.}

        \begin{lemma}\label{l:CCS:typ:struct}
    The type system enjoys
    invariance under $\equiv$ and subject reduction:
    $(i)$
  If \typ{\G;\Rs}P and $P\equiv P'$, then \typ{\G;\Rs}{P'}.
  $(ii)$
  If  \typ{\G;\Rs}P and $P\stra{} P'$, then \typ{\G;\Rs}{P'} and
  $\flocks{P'}=\flocks P$. 
\end{lemma}

\paragraph*{Deadlock-Freedom.}
  Intuitively, a deadlock in \CCSL{} arises from an acquire operation
  that cannot be performed.
    We say that a \emph{terminated process} is a parallel composition of release
  operations possibly under some restrictions.
    A process that contains at least an acquire and cannot
  reduce is a \emph{stuck process}. So in particular $\lel.\outC\lel$
  is stuck; the context may provide a release of \lel, triggering the
  acquire on \lel. On the other hand, if $P$ is a stuck process and
  complete, then $P$ is \emph{deadlocked}: intuitively, the context cannot
  interact with $P$ in order to trigger an acquire operation of $P$.
  Process \Pdl{} from Section~\ref{s:intro} is an example of a deadlock.
    We show that a complete process can only reduce to
  a terminated process, avoiding deadlocks.

   The proof of deadlock-freedom for \CCSL{} provides the structure of
the proofs for deadlock-freedom in \piL{} and leak-freedom in \piw. It
relies on progress: any typable process can reduce to reach a
terminated process.
We first present some lemmas related to the absence of cyclic
structures in \CCSL.

\begin{lemma}[Lock-connected processes]\label{l:lock:connect}
 We say  that $P$ is \emph{lock-connected} if
 \typ{\G;\Rs}P implies $\G=\G',\g$ for some $\G',\g$, with $\flocks
 P\subseteq\g$.
 In this situation,  we also have \typ{\set\g;\Rs}P.
  If $P$ and $Q$ are lock-connected and $\flocks P\cap\flocks Q$
  contains at least two distinct names, then $P|Q$ cannot be typed.
\end{lemma}
The property in Lemma~\ref{l:lock:connect} does not hold if $P$ and
$Q$ are not lock-connected: take for instance
$P=Q=\lel_1.\outC{\lel_1}|\lel_2.\outC{\lel_2}$, then we can derive
\typ{\set{\set{\lel_1},\set{\lel_2}};\emps}{P|Q}.
By the typing rule \trans{Acq-C}, any process of the form $\lel.P$ is
lock-connected. A typical example of a lock-connected process is
$\ell_1.(\outC{\lel_1}|\outC{\ell_2})|\ell_2.(\outC{\lel_2}|\outC{\ell_3})$: here
$\g=\set{\ell_1,\ell_2,\ell_3}$. Processes similar to this one are
used in the following lemma.

\begin{lemma}[No cycle]\label{l:nocycle}
  We write $P\sharel{\ell} Q$ when $\ell\in\flocks P\cap\flocks Q$. 
  Suppose there are $k>1$ pairwise distinct names
  $\ell_1,\dots,\ell_{k}$, and 
  processes $P_1,\dots,P_k$ such that
  $\ell_i.P_i\sharel{\ell_i}\ell_{(i+1)\!\mod k}.P_{(i+1)\!\mod k}$ for $1\leq i\leq k$.
  Then $P_1|\dots|P_k$ is not typable.
\end{lemma}

We use notation $\prod_iP_i$ for the parallel composition of processes $P_i$.
\begin{lemma}[Progress]\label{l:CCSL:prog}
  If \typ{\G;\flocks P}P, then either $P\stra{}P'$ for some $P'$, or 
  $P\equiv(\new\many\ell)\prod_i\outC{\ell_i}$ where     the $\ell_i$s are pairwise
  distinct. 
          \end{lemma}
\begin{proof}
  Write
  $P\equiv(\new\many\ell)P_0
  \text{ with }P_0 =
  \prod_i\outC{\ell_i}~|~\prod_j\lel_j.P_j$. We let
  $Q_j$ stand for $\lel_j.P_j$, and suppose that there is at least one
  $Q_j$.
  We show that  under
  this hypothesis $P_0$ can
  reduce.

  If $\lel_i=\lel_j$ for some $i,j$, then $P_0$ can reduce. We suppose
  in the following
  that this is not the case, and
  consider one of the $Q_j$s. By typing, there exists a unique occurrence of
  $\ell_j$ available in $P_0$. By hypothesis, this occurrence is not
  among the $\outC{\ell_i}$s.     Therefore,
  $\ell_j$ is available in $Q_{j'}$ for some unique $j'$ with $j\neq
  j'$.

  We construct a graph having one vertex for each of the $Q_j$s. We draw an
  edge between $Q_j$ and $Q_{j'}$ when  $\ell_j$ is
  available in $Q_{j'}$. By the reasoning we just made, each vertex is related to at
  least one other vertex. So the graph necessarily contains a cycle.
  We can apply Lemma~\ref{l:nocycle} to derive a contradiction. 
  
  We make two remarks about the construction of the graph. First,
  two $Q_j$s may start with an acquire at the same name. The corresponding
  vertices will have edges leading to the same $Q_{j'}$, and the
  construction still works.
    Second, if there is only one $Q_j$,   then the available release of $\lel_j$ 
  can synchronise with $Q_j$.   \end{proof}

  {By Lemma~\ref{l:CCSL:prog}, we have that any typable process is not deadlocked.
  Thus, by subject reduction, we can prove deadlock-freedom.
  \begin{proposition}[Deadlock-freedom]
    If \typ{\G;\Rs}P and $P\wtra{}P'$, then $P'$ is not deadlocked.
  \end{proposition}  }

\begin{remark}\label{r:nf}
  As \CCSL{} is finite, deadlock-freedom ensures that no acquire operation waits forever in a complete typable
  process,   and every complete process reduces to a terminated process:
  if \typ{\G;\flocks P}P,
  then $P\wtra{}
  (\new\many\ell)\prod_i\outC{\ell_i}$ where 
  the $\ell_i$s are pairwise
  distinct.
\end{remark}

\subsection{\piL: Deadlock-Freedom for Higher-Order Locks}\label{s:piL:deadlock}

\paragraph*{Syntax and Operational Semantics of \piL.}

 \piL{} extends \CCSL{} with the possibility to store \emph{values},
 which can be either booleans or locks, in locks. In this sense,
 \piL{} features higher-order locks.
 Processes 
 in \piL{}
 are defined as follows:
 \begin{mathpar}
   P\quad::=\quad \inp\lel{\lel'}.P\OR \out\ell v\OR (\new\ell)P\OR
   P_1|P_2\OR \nil\OR [v=v']P_1,P_2
   .
 \end{mathpar}
$v,v'$ denote \emph{values}, defined by    
$v~::=~
   \lel\OR\bval$, where $   \bval~::=~\ttrue\OR\ffalse$ is a boolean
   value.
  In addition to $\lel, \lel'\dots$,
we sometimes use also $x,y\dots$ to range over lock names, to suggest
a specific usage, like, e.g.\ in $\inp\lel x.P$.

 Process \out\lel{\lel'} is a \emph{release of \lel}, and
 $\inp\lel{\lel'}.P$ is an \emph{acquire on \lel}; we say in both
 cases that \lel{} is the \emph{subject} (or that \lel{} occurs in
 subject position) and $\lel'$ is the \emph{object}.
    Restriction and the acquire prefix act as binders, giving rise to the
notion of bound and free names. As in \CCSL, we write \flocks P for the set of free
lock names of $P$.
$P\sub v\lel$ is the process obtained by replacing
every free occurrence of \lel{} with $v$ in $P$.
We say that an occurrence of a process $Q$ in $P$ is \emph{guarded}
if the occurrence is under an acquire prefix, otherwise it is said
\emph{at top-level} in $P$. Additional operators w.r.t.\ \CCSL{}
 are the inactive process, \nil, and value comparison:
 $[v=v']P_1,P_2$  behaves like $P_1$ if values $v$ and $v'$
 are equal, and like $P_2$ otherwise.

Structural congruence in \piL{} is defined by adding the
 following
 axioms to $\equiv$
 in \CCSL: 
\begin{mathpar}[]
  P|\nil\,\equiv\,P

  (\new\lel)\nil\,\equiv\,\nil

  [v=v]P_1,P_2\,\equiv\, P_1

  [v=v']P_1,P_2\,\equiv\, P_2\text{ if }v\neq v'
\end{mathpar}
The last axiom above cannot be used under an acquire prefix: see
Appendix~\ref{a:piL} for the definition of $\equiv$.
\emph{Execution contexts}, 
are defined by $E
~::=~ \hole\OR E|P\OR (\new\lel)E$.
The axiom for reduction in \piL{} is:
\begin{mathpar}
  \inferrule{ }{\out\lel v~|~\inp\lel{\lel'}.P ~\red~
    P\sub v{\lel'} }
\end{mathpar}

\wtra{} is  the reflexive transitive closure of \stra{}.
  Labelled transitions, written $P\stra\mu P'$, use actions $\mu$
defined by $\mu~::=~\inp\ell v\OR\out\ell v\OR \bout\ell{\ell'}\OR\tau$, and are standard~\cite{SanWal}---we
recall the definition in Appendix~\ref{a:piL}.

\paragraph*{The type system.}

We enforce a sorting discipline for names~\cite{Mil91}, given by $V ::= \bbool\OR L$ and $\Sigma(L)=V$: values, that are
stored in locks, are either booleans or locks.  We consider that all
processes we write obey this discipline, which is left implicit. This
means for instance that when writing \out\ell v, $\ell$ and $v$ have
appropriate sorts; and similarly for $\ell(\lel').P$. In
$[v=v']P_1,P_2$, we only compare values with the same sort.

 \begin{figure}[t]
\begin{mathpar}
          \inferrule[{Acq}]{
    \typ{\G,({\g,\ell,\lel'}); \Rs,\ell}P
  }{ \typ{\set{\Gapp(\G) \uplus (\g,\ell)};\Rs}{\inp\ell{\lel'}.P}} 
  
  \inferrule[{Rel}]{ 
  }{ \typ{\G,(\g,\lel,v); \set\ell}{\out{\ell}{v}} }   
  
  \inferrule[{New}]{\typ{\G,(\g,\ell);\Rs,\ell}P
  }{\typ{\G,\g;\Rs}{(\new \ell)P}} 
  \\
  \inferrule[{Par}]{\typ{\G_1;\Rs_1}{P_1}\and \typ{\G_2;\Rs_2}{P_2} }{
    \typ{\Gop{\G_1}{\G_2};\Rs_1\uplus \Rs_2}{ P_1|P_2} }
  
  \inferrule[{Mat}]{\typ{\G;\Rs}{P_1}    \and\typ{\G;\Rs}{P_2}
  }{\typ{\G;\Rs}{[v=v']P_1,P_2} }       
  
\end{mathpar}
   \caption{Typing rules for \piL}
   \label{f:piL}
 \end{figure}

The typing judgement is written \typ{\G;\Rs}P, where \G{} and $\Rs$ are
defined like for \CCSL. We adopt the convention that if $v$ is a
boolean value, then $\g,v$ is just \g, and similarly, $\g,\lel$ is
just \g{} if the sort of \lel{} is \bbool.
The operation \Gop{\G_1}{\G_2} is
the same as for \CCSL.

 The typing rules for \piL{} are presented in Figure~\ref{f:piL}.
 Again, in rules \trans{Acq} and \trans{New}, writing $\Rs,\lel$
  imposes $\ell\notin \Rs$, otherwise the rule cannot
 be applied. Similarly, the notation $\g,\ell,\ell'$ is only defined
 when $\g\#\set{\ell,\ell'}$ and $\ell\neq\ell'$.
 Rule \trans{Rel} describes the release of a lock containing either a
 lock or a boolean value: in the latter case, using the convention
 above, the conclusion of the rule is \typ{\set{\set{\lel}};\set\lel}{\out\lel\bval}.
In rules \trans{Acq} and \trans{Rel}, the subject and the object of
the operation should belong to the same component. In
\CCSL, only prefixing yields such a constraint.

In rule \trans{Mat}, we do not impose $\set{v,v'}\in\dom\G$.
A typical example of a process that uses name comparison is
$[\lel_1=\lel_2]\out{\lel}\ttrue,\out{\lel}\ffalse$: in this process,
$\lel_1$ and $\lel_2$ intuitively represent no threat of a deadlock. 

\medskip

Before presenting the properties of the type system, we make some
comments on the discipline it imposes on $\pi$-calculus names when
they are used as locks.

\begin{remark}[An acquired lock cannot be stored]\label{r:acquired}
In \piL, the obligation to release a lock cannot be
transmitted. Accordingly, $\lel'\notin\Rs=\set\lel$ in rule
\trans{Rel}, and a process like $\lel(\lel').\out{\lel_1}{\lel}$
cannot be typed.
We return to this point after Proposition~\ref{p:lock}.
\end{remark}

\begin{remark}[Typability of higher-order locks]  
Locks are a particular kind of names of the asynchronous
$\pi$-calculus (\Api). Acquiring
a lock that has been stored in another lock boils down to
performing a communication in \Api.
We discuss how such communications can occur between typed processes.

In rule \trans{Rel}, $\ell$ and $\ell'$ must belong to the same component
of \G. So intuitively,  if  a process contains \out\lel{\lel'}, this
release is the only place where these locks are used `together'.
A reduction involving a well-typed process containing this release therefore looks like
  \begin{mathpar}
    (\out\ell{\ell'}| P)~~|~~(\ell(x).Q|Q') \quad\red\quad
    P\,|\,Q\sub{\ell'}x\,|\,Q'
    .
  \end{mathpar}
    Parentheses are used to suggest an interaction between two
  processes; $\out\lel{\lel'}|P$ performs the release, and
  $\lel(x).Q|Q'$ performs the acquire.
    Process $P$, which intuitively is the continuation of the release, may use
  locks \lel{} and $\lel'$, but not together, and similarly for $Q'$. 
  For instance we may have $P= P_\lel|P_{\lel'}$, where $\lel'$ does not
occur in $P_\lel$, and vice-versa for $P_{\lel'}$.
  Note also that $\lel'$ is necessarily fresh for
  $\lel(x).Q'$: otherwise, typability of $\lel(x).Q'$ would
  impose $\lel$ and $\lel'$ to be in the same component, which
  would forbid the parallel composition with \out\lel{\lel'}.

  Depending on how $P, Q$ and $Q'$ are written, we can envisage several patterns of
  usages of locks \lel{} and $\lel'$.
    A first example is ownership transfer (or delegation):
  $\lel'\notin\flocks P$, that is, $P$ renounces usage of
  $\lel'$. $\lel'$ can be used in $Q$. Note that typing actually also allows
  $\lel'\in\flocks{Q'}$, i.e., the recipient already knows
  $\lel'$.

  A second possibility could be that \lel{} is used linearly, in the
  sense that there is exactly one acquire on \lel. In this case, we
  necessarily have $\lel\notin\flocks P\cup\flocks{Q'}$---note that a
  release of \lel{} is available in $Q$, by typing. 
  Linearity of \lel{} means here that exactly one interaction takes place at
  \lel. After that interaction, the release on \lel{} contained in $Q$
  is \emph{inert}, in the sense that no acquire can synchronise with it.
  {We believe that this form of linearity can be used to encode binary
    session types in an extended version of \piL, including variants
    and polyadicity, along the lines of~\cite{koba:survey:ext,DBLP:journals/iandc/DardhaGS17,DBLP:conf/ppdp/DardhaGS22}.}
\end{remark}

The type system for \piL{} satisfies the same properties as in \CCSL{}
(Lemma~\ref{l:CCS:typ:struct}):
invariance under
structural congruence, merging components and subject reduction.
We also have progress and deadlock-freedom:

\begin{lemma}[Progress]\label{l:piL:progress}
  Suppose  \typ{\G;\flocks P}P, and $P$ is not structurally equivalent to
  \nil. Then
  \begin{itemize}
  \item either there exists $P'$ such that  $P\red P'$,
          \item or $P\equiv(\new\many{\ell})(\Pi_i \outC{\ell_i}v_i)$         where the $\ell_i$s are pairwise
    distinct.   \end{itemize}
\end{lemma}

Like in \CCSL, a  deadlocked process in \piL{} is defined as a complete
  process that is stuck.

\begin{proposition}[Deadlock-freedom]\label{p:lock}
  If \typ{\G;\Rs}P and $P\wtra{}P'$, then $P'$ is not deadlocked.
      \end{proposition}

The proof of deadlock-freedom is basically the same as for \CCSL. The
reason for that is that although the object part of releases plays a
role in the typing rules, it is not relevant to establish
progress (Lemma~\ref{l:piL:progress}). This is the case because in \piL, it is not possible to store
an acquired lock in another lock (Remark~\ref{r:acquired}).

      It seems difficult to extend the type system in order to allow
processes that transmit the release obligation on a lock. This would
make it possible to type-check, e.g., process
$\lel(\lel').\out{\lel_1}\lel$, that does not release lock
$\lel$ but instead stores it in $\lel_1$. Symmetrically, a process
accessing $\lel$ at $\lel_1$ would be in charge of releasing both
$\lel_1$ and $\lel$.
In such a framework, a process like
$\res{\lel_1}(\out{\lel_1}{\lel}|\lel(x).\out{\lel}x)$ would be deadlocked,
because the inert release \out{\lel_1}{\lel} contains the release
obligation on $\lel_1$. 
The type system in Section~\ref{s:piw} makes it possible to transmit
the obligation to perform a release (and similarly for a wait).

      \begin{remark}
Similarly to Remark~\ref{r:nf}, we have that \typ{\G;\flocks
  P}P implies $P\wtra{}
(\new\many{\ell})(\Pi_i \outC{\ell_i}v_i)$
    where the $\ell_i$s are pairwise
    distinct.
As a consequence, the following holds: if \typ{\G;\flocks P}P,
then for any $\lel\in\flocks P$, $P\wtra{}\stra\mu$, where $\mu$ is a
release of \lel.
This statement would be better suited if infinite computations were
possible in \piL. We leave the investigation of such an extension of
\piL{} for future work.
\end{remark}

\subsection{Behavioural Equivalence in \piL}\label{s:piL:behav}

We introduce typed barbed equivalence (\wbe) and
typed bisimilarity (\bisL) for \piL. We show that \bisL{} is a sound technique
to establish \wbe{}, and present several examples of
(in)equivalences between \piL{} processes.

\subsubsection{Barbed Equivalence and Labelled Semantics for \piL}

A typed relation in \piL{} is a set of quadruples of the form $(\G,\Rs,P,Q)$
such that \typ{\G;\Rs}P and \typ{\G;\Rs}Q. When a typed relation \R{}
contains $(\G,\Rs,P,Q)$, we write \typ{\G;\Rs}{P\R Q}.
We say that a typed relation \R{} is symmetric if \typ{\G;\Rs}{P\R Q} implies
\typ{\G;\Rs}{Q\R P}.

Deadlock-freedom has two consequences
regarding the definition of barbed equivalence in \piL, noted
\wbe.
First, only complete processes should be observed, because intuitively a
computation in \piL{} should not be blocked by an acquire operation
that cannot be executed.

Second, Proposition~\ref{p:lock} entails that all weak barbs in the sense of \Api{}
can always be observed in \piL.
In \Api, a weak barb at $n$ corresponds to the possibility
 to reduce to a process
in which an output at channel
$n$ occurs at top-level.
We  need a stronger notion of barb, otherwise \wbe{} would be
trivial.
That behavioural equivalence in \piL{} is not trivial is shown for
instance by the presence of non-determinism. 
Consider indeed process
 $P_c\defi(\new\lel)\big(\, \lel(x).(\out cx|\out\lel x)~|~
  \lel(y).\out\lel\ffalse~|~ \out\lel\ttrue\,\big)$. 
  Then $P_c\wtra{}\out c\ttrue$ and $P_c\wtra{}\out c\ffalse$ (up to the
  cancellation of an inert process of the form
  $(\new\lel)\out\lel\bval$). 
We therefore  include the object part of
releases in barbs. 
We write $P\barb{\out\ell{\ell'}}$ if $P\stra{\out\ell{\ell'}}$, and
$P\barb{\bout\ell\new}$ if 
$P\stra{\bout\ell{\ell'}}$.
We use $\eta$ to range over barbs,  writing $P\barb\eta$; 
the weak version of the predicate, defined as $\wtra{}\barb\eta$, is written
$P\wbarb\eta$.

\begin{definition}[Barbed equivalence in \piL, \wbe]\label{d:wbe}
  A symmetric typed relation \R{} is a \emph{typed barbed bisimulation} if
  \typ{\G;\Rs}{P\R Q} implies the three following properties:
  \begin{enumerate}
  \item whenever
    $P, Q$ are complete and
    $P\stra{}P'$, there is $Q'$ s.t.\
    $Q\wtra{}Q'$ and \typ{\G;\Rs}{P'\R Q'};
  \item for any $\eta$, if $P, Q$ are complete and
    $P\barb\eta$ then 
   $Q\wbarb\eta$;
 \item for any  $E,\G',\Rs'$ s.t. \typ{\G';\Rs'}{E[P]} and
   \typ{\G';\Rs'}{E[Q]},
 and $E[P], E[Q]$ are complete,
 we have \typ{\G';\Rs'}{E[P]\,\R\, E[Q]}.
  \end{enumerate}
  \emph{Typed barbed equivalence}, written \wbe{}, is the greatest typed
  barbed bisimulation.
\end{definition}

\begin{lemma}[Observing only booleans]\label{l:obs:bool}
  We use $\obs,\obs',\dots$ for lock names that are used to store
  boolean values. We define \wbeo{} as the equivalence defined as in
  Definition~\ref{d:wbe}, but restricting the second clause to barbs
  of the form
  \barb{\out\obs\bval} and
  \wbarb{\out\obs\bval}.
    Relation \wbeo{} coincides with \wbe.   
\end{lemma}

To define typed bisimilarity,
we introduce \emph{type-allowed transitions}. The terminology
means that we select among the untyped transitions those that are
fireable given the constraints imposed by types.

\begin{definition}[Type-allowed transitions]\label{d:typ:trans:piL}
	When $\typ{\Gamma;\Rs}{P}$, we write
        $\rconf{\Gamma;\Rs}{P}\stra{\mu}\rconf{\Gamma';\Rs'}{P'}$ 
	if $P\stra{\mu}P'$ and one of the following holds:
	\begin{enumerate}
        \item $\mu=\tau$, in which case $\Rs'=\Rs$ and $\G'=\G$;

        \item $\mu=\out\ell v$, in which case $(\g,\ell,v)\in\G$
          for some \g, and  $\Rs',\ell=\Rs$, $\G'=\G$;

        \item\label{it:typ:trans:bout} $\mu=\bout\ell{\ell'}$, in
          which case
          $\G = \G_0,(\g,\ell)$ for some $\G_0,\g$, 
          $\G' = \G_0, (\g,\ell,\ell')$, and we have
           $\Rs',\ell=\Rs,\ell'$;
         \item\label{typed:trans:acq} $\mu=\inp\ell{v}$, in which case there are
           $\G_0,\Rs_0$ s.t.\ \typ{\G_0;\Rs_0}{P|\out\lel{v}}, and
           $\G'=\G_0, \Rs'=\Rs_0$.
        \end{enumerate}
        \end{definition}

        In
        item~\ref{it:typ:trans:bout},
        $\ell$
        is removed from the \Rs{} component, and $\lel'$ is added:
it is $P'$'s duty to perform the release of  $\lel'$, the obligation
is not transmitted.
An acquire transition involving a higher-order lock
  merges two distinct components in the typing environment:
  if
  $[\G_0,(\g,\ell),(\g',\ell');\Rs;P]\stra{\inp\lel{\lel'}}[\G';\Rs';P']$
  (item~\ref{typed:trans:acq} above),
  then
  $\G' = \G_0,(\g\uplus\g'\uplus\set{\ell,\ell'})$ and 
          $\Rs'=\Rs,\ell$ (and in particular  $\lel\notin \Rs$).

\begin{lemma}[Subject Reduction for type-allowed transitions]\label{l:SR}
  If $[\G;\Rs;P]\stra\mu[\G';\Rs';P']$, then \typ{\G';\Rs'}{P'}.
\end{lemma}

\begin{definition}[Typed bisimilarity, \bisL]\label{d:bisL}
  A typed relation \R{} is a typed bisimulation if \typ{\G;\Rs}{P\R Q}
  implies that whenever $[\G;\Rs;P]\stra\mu[\G';\Rs';P']$, we have
  \begin{enumerate}
  \item either  $Q\wtra{\hat\mu}Q'$ and
    \typ{\G';\Rs'}{P'\R Q'} for some $Q'$
      \item or $\mu$ is an acquire $\ell(v)$,
    $Q|\out\ell{v}\wtra{} Q'$ and
    \typ{\G';\Rs'}{P'\R Q'} for some $Q'$,
  \end{enumerate}
  and symmetrically for the type-allowed transitions of $Q$.

\emph{Typed bisimilarity}, written \bisL, is the largest typed
bisimulation.
\end{definition}
We write \typ{\G;\R}{P\bisL Q} when $(\G;\Rs,P,Q)\in\,\bisL$.
If \typ{\G;\Rs}{P\bisL Q} does not hold, we write \typ{\G;\Rs}{P\nbisL
  Q}, and similarly for \typ{\G;\Rs}{P\nwbe Q}.

Proposition~\ref{p:piL:sound} below states that relation \bisL{} provides a sound proof technique for \wbe. The main
property to establish this result is that \bisL{} is preserved by
parallel composition:
 \typ{\G_0;\Rs_0}{P\bisL Q} implies that for all $T$, whenever
  \typ{\G;\Rs}{P|T} and \typ{\G;\Rs}{Q|T}, we have
  \typ{\G;\Rs}{P|T\bisL Q|T}.

\begin{proposition}[Soundness]\label{p:piL:sound}
  For any $\G,\Rs,P,Q$, if \typ{\G;\Rs}{P\bisL Q}, then 
  \typ{\G;\Rs}{P\wbe Q}.
\end{proposition}

The main advantage in using \bisL{} to establish equivalences for
\wbe{} is that we can reason directly on processes, even if they are
not complete.

\subsubsection{Examples of Behavioural Equivalence in \piL}

\begin{example} We discuss some equivalences for \wbe.

The equivalence \typ{\set{\set\lel};\emps}{\lel(x).\out\lel x\wbe\nil}, which is
typical of \Api, holds in
 \piL.
This follows
directly from the definition of typed bisimilarity, and soundness (Proposition~\ref{p:piL:sound}).

\smallskip

We now let
$P\defi \inp\lel x.(\out{\lel_0}\ttrue|\out\lel x)$ and
  $Q\defi \out{\lel_0}\ttrue$, and consider whether we can
  detect the presence of a  `forwarder' at \lel{} when its behaviour is interleaved with another
  process. $P$ and $Q$ have different barbs---they are obviously not complete.
        It turns out that \typ{\set{\set{\lel,\lel_0}};\set{\lel_0}}{\inp\lel
  x.(\out{\lel_0}\ttrue|\out\lel x)\nwbe\out{\lel_0}\ttrue}. Indeed,
let us consider the context
  $$
  E\, \defi\, \hole~|~ \lel_0(y).\inp w\_.(\out w\ttrue|\out{\lel_0}y)~|~
  \inp{w'}\_.(\out{w'}\ttrue|\out\lel v)~|~\out
  w\ffalse|\out{w'}\ffalse
  ,
  $$
  where $w,w'$ are fresh names and $v$ is a value of the appropriate sort.
  We have $E[Q]\wtra{}Q'$ with $Q'\not\barb{\out w\ffalse}$ and
  $Q'\barb{\out{w'}\ffalse}$. On the other hand, for any $P'$ s.t.\
  $E[P]\wtra{}P'$, if $P'\not\barb{\out w\ffalse}$, then
  $P'\not\barb{\out{w'}\ffalse}$.

  Contexts like $E$ above make it possible to detect when the process in
  the hole has some interaction (here, with locks \lel{} and $\lel_0$).

  Using similar ideas, we can prove that
  $$\typ{\G;\Rs}{\inp{\lel_1}x.\inp{\lel_2}y.P ~\nwbe~
    \inp{\lel_2}y.\inp{\lel_1}x.P}
  \qquad\text{ for appropriate \G{} and \Rs.}$$
  Indeed, let us define $E_w\, \defi\, \out w\ffalse~|~\inp w{\_}.(\hole|\out
  w\ttrue)$, where $\_$ stands for an arbitrary lock name, that is not
  used. We can use the context
  $\hole~|~ E_{w_2}[\out{\ell_2}{v_2}]~|~ \out{\ell_1}{v_1}
  ~|~ \ell_1(z).E_{w_1}[\out{\ell_1}z]$, for fresh names $w_1,w_2$ and
  appropriate values $v_1, v_2$, to detect
    the order in which acquires on $\lel_1$ and $\lel_2$
  are made.
              \end{example}

In the next two examples, we show equivalences that hold because we
work in a typed setting.
\begin{example}
  Suppose \typ{\G;\Rs,\lel}{\inp\ell x.P~|~ \inp{\ell'} y.(\out\ell v~|~
    Q)}. Then we have
  \begin{mathpar}
    \typ{\G;\Rs,\lel}{}\inp\ell x.P~|~ \inp{\ell'} y.(\out\ell v~|~ Q)
    ~\bisL~
    \inp{\ell'}y.(\out\ell v~|~Q~|~ \inp\ell x.P),
  \end{mathpar}
  \noindent because intuitively the acquire on $\ell$ cannot be
  triggered by the context, due to the presence of a release at
  \lel{} in the process.
  (We  remark in passing that \typ{\G;\Rs,\lel}{
\inp\ell x.P~|~ \inp{\ell'} y.(\out\ell v~|~ Q)}
iff \typ{\G;\Rs,\lel}{
    \inp{\ell'}y.(\out\ell v~|~Q~|~ \inp\ell x.P)}, and in this case
  \G{} contains a component of the form $(\g,\lel,\lel',v)$.)

  \medskip
  
  This law can be generalised as follows. We say that $\ell$ is
  \emph{available} in a context $C$ if the hole does not occur in $C$
  neither under a binder for \lel, nor under an acquire on \lel.
      So for instance $\ell$ is not available in $(\new\ell)\hole$, in
  $\inp{\ell_0}\ell.\hole$ or in $\inp\ell x.\hole$, and $\ell$ is available
  in $\inp\ell x.\out\ell x~|~\out\ell v~|~\hole$.
If \lel{} is available in $C$, then
  \begin{mathpar}
    \typ{\G;\Rs}{}\inp\ell x.P~|~C[\out\ell v] \bisL
    C[\out\ell v~|~\inp\ell x.P]
    \qquad\text{ for appropriate \G{} and \Rs.}
  \end{mathpar}
\end{example}

\begin{example}
  Consider the following processes:
    \begin{mathpar}
    \begin{array}{l}
    P_1~~=~~
    (\new\ell_1)\big(\,\ell_1.\ell_2.(\outC{\ell_1}|\outC{\ell_2})
    ~|~ \ell(x).\ell_1.x.(\outC{\ell_1}|\outC x|\out\lel x)~|~
    \outC{\ell_1}|\outC{\ell_2}
    \,\big)
      \\
    P_2~~=~~
    (\new\ell_1)\big(\,\ell_1.\ell_2.(\outC{\ell_1}|\outC{\ell_2})
    ~|~ \ell(x).x.\ell_1.(\outC{\ell_1}|\outC x|\out\lel x)~|~
    \outC{\ell_1}|\outC{\ell_2}
    \,\big)
    \end{array}
  \end{mathpar}
  Here we use a CCS-like syntax,
    to ease readability. This notation means that acquire operations are
  used as forwarders, i.e., the first component of $P_1$ and $P_2$
  should be read as
  $\lel_1(y_1).\lel_2(y_2).(\out{\lel_1}{y_1}|\out{\lel_2}{y_2})$.
    Moreover, the two
  releases available at top-level are
  $\out{\lel_1}\ttrue|\out{\lel_2}\ttrue$, and similarly for \out x\ttrue{} (the reasoning also holds if
  $\lel_1$ and $\lel_2$ are higher-order locks).

  In the pure $\pi$-calculus, $P_1$ and $P_2$ are not equivalent,
  because $\lel_2$ can instantiate $x$ in the acquire on \lel.
  We can show \typ{\set{\set{\lel_2,\lel}};\set{\lel_2}}{P_1\bisL P_2} in \piL,
  because the transition \stra{\inp\lel{\lel_2}} is ruled out by the
  type system.
\end{example}

\section{\piw, a Leak-Free Asynchronous $\pi$-Calculus}\label{s:piw}

\subsection{Adding Lock Deallocation}\label{s:def:piw}

\piw{} is obtained from \piL{} by adding the \emph{wait construct}
$\wait\lel{\lel'}.P$ to the grammar of \piL. 
As announced in Section~\ref{s:intro}, the following reduction rule
describes how wait interacts with a release:
\begin{mathpar}
  \inferrule{ }{(\new\lel)(\out\lel v~|~\wait\lel{\lel'}.P) ~\red~
    P\sub v{\lel'}}
  ~\lel\notin\flocks P
\end{mathpar}
The wait instruction deallocates the lock. The continuation may use
$\lel'$, the final value of the lock. We say that \wait\lel{\lel'} is a \emph{wait
on \lel}, and $\lel'$ is bound in $\wait\lel{\lel'}.P$.

Types in \piw, written $\ttyp,\ttyp',\dots$, are defined by
$\ttyp~::=~ \bbool\OR \tlockw\ttyp{rw}$, 
and \emph{typing hypotheses} are written $\lel:\ttyp$.
In $\lel:\tlockw\ttyp{rw}$, $rw$ is called the \emph{usage} of \lel, and
$r,w\in\set{0,1}$ are the release and wait \emph{obligations}, respectively, on lock \lel.
So for instance a typing hypothesis of the form
$\lel:\tlockw\ttyp{10}$ means that \lel{} must be used to perform a
release and cannot be used to perform a wait.
An hypothesis $\lel:\tlockw\ttyp{00}$ means that \lel{} can only be used to perform
acquire operations.
This  structure for types makes it possible  to transmit
the wait and release obligations on a given lock name via
higher-order locks.

Our type system ensures that locks are properly deallocated.
In contrast to \piL, this allows acquired locks to be stored without
creating deadlocks. For example, a process like $\res{\lel_1}(\out{\lel_1}{\lel}|\lel(x).\out{\lel}x)$
is deadlocked if $\lel_1$ stores the release obligation of $\lel$;
however, it cannot be typed
as it lacks the wait on $\lel_1$. Adding a wait, e.g. $\wait{\lel_1}{\lel}.\out{\lel}v$
removes the deadlock.

Typing environments have the same structure as in Section~\ref{s:piL},
except that components \g{} are sets of {typing hypotheses}
instead of simply sets of lock names.
\dom\G{} is defined as the set of lock names for which \G{} contains a
typing hypothesis.
We write $\G(\lel)=\ttyp$ if the typing hypothesis
$\lel:\ttyp$ occurs in \G.

We reuse the notation for composition of typing environments. \Gop{\G_1}{\G_2} is defined like in
Section~\ref{s:CCSL}, using the \connname{} operator, to avoid
cyclic structures in the sharing of lock names.
Additionally, when merging components, we  compose typing
hypotheses. For any \lel, if
$\lel:\tlockw{\ttyp_1}{r_1w_1}\in\dom{\G_1}$ and
$\lel:\tlockw{\ttyp_2}{r_2w_2}\in\dom{\G_2}$, the typing hypothesis
for \lel{} in
\Gop{\G_1}{\G_2} is $\lel:\tlockw\ttyp{(r_1+r_2)(w_1+w_2)}$, and is
defined only if $\ttyp=\ttyp_1=\ttyp_2$, 
$r_1+r_2\leq 1$ and $w_1+w_2\leq 1$.

\begin{figure}[t]
  \centering
\begin{mathpar}
    \inferrule[Acq-w]{ \typ{\G,(\g, \lel:\tlock\ttyp_{1w},\lel':\ttyp) }P
  }{ \typ{\set{\Gapp(\G) \uplus (\g,\lel:\tlock\ttyp_{0w})}}{\inp\lel{\lel'}.P} }

  \inferrule[Rel-w]{   }{
    \typ{ \G_{00},(\g_{00},\lel:\tlock{\ttyp}_{10},v:\ttyp)
        }{\out\lel{v}}
    }
    
  \inferrule[Wait-w]{ \typ{\set{\g,\lel':\ttyp}} P }{
    \typ{\set{\g,\lel:\tlock\ttyp_{01}}}{\lel((\lel')).P} }

  \inferrule[New-w]{\typ{\G,(\g,\lel:\tlock{\ttyp}_{11})}P }{
    \typ{\G,\g}{(\new\lel)P} }

  \inferrule[Par-w]{\typ{\G_1}{P_1}\and \typ{\G_2}{P_2}
  }{
    \typ{\Gop{\G_1}{\G_2}}{P_1|P_2} }

  \inferrule[Nil-w]{ }{\typ{
    \emps}\nil }
  
  \inferrule[Mat-w]{\typ{\G}{P_1}     \and\typ{\G}{P_2}
  }{\typ{\G}{[v=v']P_1,P_2} }   
\end{mathpar}
\caption{Typing rules for \piw
}
  \label{f:piw}
\end{figure}

The typing rules are given in Figure~\ref{f:piw}. The rules
build on the rules for \piL, and rely on usages to
control the release and wait obligations.
In particular, the set \Rs{} in Figure~\ref{f:piL} corresponds to the
set of locks whose usage is of the form $1w$ in this system.
To type-check an acquire, we can have usage $00$, but also $01$, as
in, e.g., $\inp\lel{\lel'}.(\out\lel{\lel'}|\wait\lel x.P)$. In rule
\trans{Rel-w}, we impose that all typing hypotheses in $\G_{00}$ (resp.\ $\g_{00}$)
have the form $\ell:\tlockw{\ttyp}{00}$.

Several notions introduced for the type system of
  Section~\ref{s:piL} have to be adapted in the setting of \piw.
  While in Section~\ref{s:piL} we simply say that a lock \lel{} is
  available, here we distinguish whether a release of \lel{} or a wait
  on \lel{} is available.
If $P$ has a subterm of the form $\wait\lel x.Q$ that does not occur
under a binder for \lel, we say that a wait on \lel{} is
\emph{available} in $P$.
If \out\lel v occurs in some process $P$ and this occurrence is
neither under a binder for \lel{} nor under an acquire on \lel, we say
that a release of \lel{} is \emph{available} in $P$.
In addition, a release of \lel{} (resp.\ wait on \lel) is available in $P$ also if $P$
contains a release of the form \out{\lel_0}\lel, which does not occur under a
binder for \lel, and if $\lel$'s type is of the form
\tlockw\ttyp{1w} (resp.\ \tlockw\ttyp{r1}).

Like in \piL, a deadlocked process in \piw{} is a complete
  process that is stuck. The notion of complete process has to be
  adapted in order to take into account the specificities of \piw.
  First, the process should not be stuck just because a restriction is
  missing in order to trigger a name deallocation. Second, we must
  consider   the fact that release and wait obligations can be
  stored in locks in \piw.
    As a consequence, when defining complete processes in \piw, we
  impose some constraints on the free lock names of processes.

In \piw, we say that
  \emph{\G{} is complete} if for any $\lel\in\dom\G$,
  either $\G(\lel)=\tlockw\bbool{10}$ or
  $\G(\lel)=\tlockw{\tlockw\ttyp{00}}{10}$ for some \ttyp.
{
  To understand this definition, suppose \typ\G P with \G{}
complete. Then we have, for any free lock name \lel{} of $P$: $(i)$  the
release of $\ell$ is available in $P$; $(ii)$ this release does not
carry any 
obligation; $(iii)$ the wait on \lel{} is  \emph{not} 
available in $P$. The latter constraint means that if a $P$ contains a
wait on some lock, then this lock should be restricted.
}

The notion of leak-freedom we use is inspired
  from~\cite{DBLP:journals/pacmpl/JacobsB23}. In our setting, a
  situation where some lock
  \lel{} is released and will never be acquired again can be seen as a
form of memory leak. 
We say that $P$ \emph{leaks} $\lel$ if $P \equiv \res{\lel}(P' | 
  \out{\lel}{v})$ with $\lel\notin\flocks{P'}$. 
$P$ \emph{has
    a leak} if $P$ leaks \lel{} for some \lel, and is \emph{leak-free}
  otherwise.

\begin{lemma}[Progress]\label{l:piw:prog}
  If $\typ{\G}P$ and $\G$ is complete, then either $P\stra{}P'$ for some
  $P'$, or $P\equiv(\new\many{\ell})(\Pi_i \outC{\ell_i}v_i)$
  where the $\ell_i$s are pairwise
  distinct.
\end{lemma}

For lack of space, the proof is presented in Appendix~\ref{a:piw}.
Again, it follows the lines of the proof of
Lemma~\ref{l:CCSL:prog}. To construct a graph containing necessarily a
cycle, we associate to every acquire of the form $\lel(x).Q$ an
available release of \lel, which might occur in a release of the form
\out{\lel'}\lel, if $\lel'$ carries the release obligation. Similarly,
to every wait $\wait\lel x.Q$, we associate an available release, or,
if  a release
\out\lel v occurs at top-level, an acquire on 
\lel, that necessarily exists otherwise a reduction could be
fired. Finally, using a similar reasoning, to every release of \lel{}
at top-level, we associate a wait on \lel, or an acquire on \lel.

A consequence of Lemma~\ref{l:piw:prog} is that
$P \wred \nil$ when $\typ{\emptyset}P$.
\begin{proposition}[Deadlock- and Leak-freedom]\label{p:leak}
  \typ\G P and $P\wtra{}P'$, then $P'$ neither is deadlocked, nor has a
  leak.
\end{proposition}

\begin{corollary}\label{c:wait:bool}
Suppose \typ{\G,\g,\lel:\tlockw\bbool{10}}P, and suppose that the
usage of all names in $S = \dom{\G,\g}$ is $11$. Then
$(\new S)P\wbarb{\out\lel\bval}$ for some \bval.
\end{corollary}
\begin{proof}
  Immediate by Lemma~\ref{l:piw:prog} and subject reduction.
        \end{proof}
  This property is used to define barbed equivalence below.
It does not hold for higher-order locks: simply discarding $x$, the
  lock stored in \lel, might break typability if \lel{} carries
  an obligation.   
  \subsection{Typed Behavioural Equivalence in \piw}\label{s:piw:equiv}

\subsubsection{Barbed Equivalence}

In barbed equivalence in \piL{} (Definition~\ref{d:wbe}), we compare
complete \piL{} processes, intuitively to prevent blocked acquire operations
from making certain observations impossible.
Similarly, in \piw, we must also make sure that all wait operations in
the processes being observed will eventually be fired. For this, we
need to make the process complete (in the sense of
Lemma~\ref{l:piw:prog}), and to 
add restrictions so that wait transitions are fireable.

However, in order to be able to observe some barbs and discriminate
processes, we rely on Corollary~\ref{c:wait:bool}, and allow names to
be unrestricted as long as their type is of the form
\tlockw\bbool{10}. This type means that the lock is first order, and
that the context has the wait obligation. In such a situation,
interactions at \lel{} will never be blocked, the whole process is
deadlock-free, and eventually reduces to a parallel composition of
releases typed with \tlockw\bbool{10}.
Accordingly, we say that a \piw{} process $P$ is \emph{wait-closed} if \typ\G P and for any
$\lel\in\dom\G$, $\G(\lel)=\tlockw\bbool{10}$.

A typed relation in \piw{} is a set of triples $(\G,P,Q)$
such that \typ\G P and \typ\G Q, and we write \typ\G{P\R Q} for
$(\G,P,Q)\in\R$.
Barbed equivalence in \piw{} is defined like \wbe{}
(Definition~\ref{d:wbe}), restricting observations to wait-closed
processes. 
\begin{definition}[Barbed equivalence in \piw, \wbew]\label{d:webw}
  A symmetric typed relation \R{} is a \emph{typed barbed bisimulation} if
  \typ{\G}{P\R Q} implies the three following properties:
  \begin{enumerate}
  \item whenever
    $P, Q$ are wait-closed and
    $P\stra{}P'$, there is $Q'$ s.t.\
    $Q\wtra{}Q'$ and \typ{\G}{P'\R Q'};
  \item if $P, Q$ are wait-closed and
    $P\barb\eta$ then 
   $Q\wbarb\eta$;
 \item for any  $E,\G'$ s.t. \typ{\G'}{E[P]} and
   \typ{\G'}{E[Q]},
 and $E[P], E[Q]$ are wait-closed,
 we have \typ{\G'}{E[P]\,\R\, E[Q]}.
  \end{enumerate}
  \emph{Typed barbed equivalence} in \piw, written \wbew{}, is the greatest typed
  barbed bisimulation.
\end{definition}

In the second clause above, $\eta$ can only be of the form
\out\lel\bval, for some boolean value \bval. Lemma~\ref{l:obs:bool}
tells us that we could proceed in the same way when defining \wbe.

\subsubsection{Typed Transitions for \piw, and Bisimilarity}

We now define a LTS for \piw.
Transitions for name deallocation are not standard in the
$\pi$-calculus. To understand how we deal with these, consider 
$\wait\lel{\lel'}.P|Q$: this process can do
\stra{\wait\lel{v}} only if $Q$ does not use \lel.
Similarly, in  $\wait\lel{\lel'}.P|\inp\lel{x}.Q| \out\lel v$, the
acquire can be fired, and the wait cannot.

Instead of selecting type-allowed transitions among the untyped
transitions like in Section~\ref{s:piL:behav}, we give an inductive definition of typed
transitions, written \tstra\G P\mu{\G'}{P'}. This allows us to use the
rules for parallel composition in order to control the absence of a
lock, when a lock deallocation is involved.
Technically, this is done by refining the definition of the operator
to compose typing
contexts.

Actions of the LTS are defined as follows:
$\mu ~::=~ \inp\lel v \OR \out\lel v \OR \bout\lel{\lel'}\OR\tau\OR
\lel((v))\OR  \wtau\lel
$.
Name \lel{} plays a particular role in transitions along
wait actions \wait\lel v{} and \emph{wait synchronisations} \wtau\lel:
since \lel{} is deallocated, we must make sure that it is not used
elsewhere in the process.
We define \Gopl{\G_1}{\mu}{\G_2} as being equal to
\Gop{\G_1}{\G_2}, with the additional constraint that 
$\lel\notin\dom{\G_1}\cup\dom{\G_2}$ when $\mu=\wait\lel v$ or
$\mu=\wtau\lel$, otherwise
\Gopl{\G_1}{\mu}{\G_2} 
is not defined.
The rules defining the LTS are given on Figure~\ref{f:LTS:piw}.
We define $\flocks{\bout\lel{\lel'}}=\flocks{\wtau\lel}=\set\lel$, and
$\flocks{\inp\lel{v}} = \flocks{\wait\lel{v}} = \flocks{\out\lel{v}} =
\set{\lel,v}$ (with the convention that $\set{\lel,v}=\set\lel$ if $v$
is a boolean value).

We  comment on the transition rules.
Rules \trans{TR}, \trans{TA} and \trans{TW} express the meaning of
usages (respectively, $01$, $0w$ and $10$).
In rule \trans{TT}, \lel{} is deallocated, and the restriction on
\lel{} is removed.
In rules \trans{TPT}, \trans{TPTB} we rely on operation \Gopl{\G_1}\mu{\G_2} to make sure that \lel{} does not
appear in both parallel components of the continuation process, and
similarly for \trans{TPP} in the case where $\mu$ involves
deallocation of \lel.

Typability is preserved by typed transitions:   if \typ\G P{} and
\tstra{\G}P\mu{\G'}{P'}, then \typ{\G'}{P'}.

\begin{figure}[t]
  \centering
\begin{mathpar}
\inferrule[TR]{ }{\tstra{\set{\lel:\tlockw\ttyp{01},v:\ttyp}}{\out\lel{v}}
{\out\lel v}\emps\nil }

\and

\inferrule[TA]{ }{\tstra{\set{\g,\lel:\tlockw\ttyp{0w}}}{\inp\lel{\lel'}.P}
  {\inp\lel v} {\set{\g\sub
      v{\lel'},\lel:\tlockw\ttyp{1w} }} {P\sub v{\lel'}}  }

\and

\inferrule[TW]{ }{
  \tstra{\set{\g,\lel:\tlockw\ttyp{10}}}{\wait\lel{\lel'}.P}{\wait\lel
    v}
  {\set{\g\sub v{\lel'}}}{P\sub v{\lel'}}  }
\and
\inferrule[TN]{
  \tstra{\G,\g,\lel:\tlockw\ttyp{11}}P\mu{P'}{\G',\g',\lel:\tlockw\ttyp{11}} 
}{
  \tstra{\G,\g}{(\new\lel)P}\mu{\G',\g'}{(\new\lel)P'}
  }~\lel\notin\flocks\mu

  \and
  
  \inferrule[TO]{
    \tstra{\G,\g,\lel':\ttyp}
    P{\out\lel{\lel'}}{P'}{\G',\g,\lel':\ttyp'}
  }{
    \tstra{\G,\g}
    {(\new\lel')P}{\bout\lel{\lel'}}{P'}{\G',\g,\lel':\ttyp'}
  }
  
  \and

  \inferrule[TT]{
    \tstra{\G,\g,\lel:\tlockw\ttyp{11}}
    P{\tau/\lel}{\G',\g}{P'}
  }{
    \tstra{\G,\g}
    {(\new\lel)P}{\tau}{\G,\g}{P'}
  }

\inferrule[TPC]{ \tstra{\G_1}P{\inp\lel v}{\G'_1}{P'} \and
  \tstra{\G_2}Q{\out\lel v}{\G'_2}{Q'} }{
  \tstra{\Gop{\G_1}{\G_2}}{P|Q}\tau{\Gop{\G'_1}{\G'_2}} {P'|Q'}
}

\inferrule[TPB]{ \tstra{\G_1}P{\inp\lel{\lel'}}{\G'_1}{P'} \and
  \tstra{\G_2}Q{\bout\lel{\lel'}}{\G'_2}{Q'} }{
  \tstra{\Gop{\G_1}{\G_2}}{P|Q}\tau{\Gop{\G'_1}{\G'_2}}{(\new\lel')(P'|Q')}
}

\inferrule[TPP]{ \tstra{\G_1}P\mu{\G'_1}{P'} \and \typ{\G_2}Q }{
  \tstra{\Gop{\G_1}{\G_2}}{P|Q}\mu{\Gopl{\G'_1}\mu{\G_2}}{P'|Q}
}

\inferrule[TPT]{ \tstra{\G_1}P{\wait\lel v}{\G'_1}{P'} \and
  \tstra{\G_2}Q{\out\lel v}{\G'_2}{Q'} }{
  \tstra{\Gop{\G_1}{\G_2}}{P|Q}{\tau/\lel}{\Gopl{\G'_1}{\tau/\lel}{\G'_2}}{P'|Q'} 
}

\inferrule[TPTB]{ \tstra{\G_1}P{\wait\lel v}{\G'_1}{P'} \and
  \tstra{\G_2}Q{\bout\lel{\lel'}}{\G'_2}{Q'} }{
  \tstra{\Gop{\G_1}{\G_2}}{P|Q}{\tau/\lel}{\Gopl{\G'_1}{\tau/\lel}{\G'_2}}{(\new\lel')(P'|Q')} 
}

\end{mathpar}
  \caption[LTS for \piw]{\piw, Typed LTS. We omit symmetric versions
    of rules involving parallel compositions}
  \label{f:LTS:piw}
\end{figure}

Bisimilarity in \piw{} takes into account the additional
transitions w.r.t.\ \piL,
and is sound for \wbew.
\begin{definition}[Typed Bisimilarity in \piw, \bisw]
  A typed relation \R{} is a  typed bisimulation if \typ{\G}{P\R Q}
  implies that whenever $\tstra{\G}P\mu{\G'}{P'}$, we have
  \begin{enumerate}
  \item either  $Q\wtra{\hat\mu}Q'$ and
    \typ{\G'}{P'\R Q'} for some $Q'$
  \item or $\mu$ is an acquire $\ell(v)$,
    $Q|\out\ell v\wtra{} Q'$ and
    \typ{\G'}{P'\R Q'} for some $Q'$,
  \item or $\mu$ is a wait \wait{\ell}v,
    $(\new\lel)(Q|\out\ell v)\wtra{}Q'$ and
    \typ{\G'}{P'\R Q'} for some $Q'$,
  \item or $\mu=\wtau\lel$,
    $(\new\lel)Q\wtra{}Q'$ and
    \typ{\G'}{P'\R Q'} for some $Q'$, 
  \end{enumerate}
  and symmetrically for the typed transitions of $Q$.
\emph{Typed bisimilarity in \piw}, written \bisw, is the largest typed
bisimulation.
  \end{definition}

  \begin{proposition}[Soundness]\label{p:piw:sound}
    For any $\G, P, Q$, if \typ\G{P\bisw Q}, then \typ\G{P\wbew Q}.
  \end{proposition}

\begin{example}
  The law $\lel(x).\out\lel x = \nil$ holds in \piw, at type
  $\lel:\tlockw\ttyp{00}$, for any \ttyp.

  Suppose \typ\G
  {\inp\lel x.P|\wait\lel y.Q}. Then we can prove
$
    \typ\G{\inp\lel x.P|\wait\lel y.Q  ~\bisw~ \inp\lel x.(P~|~\wait\lel y.Q)}.
$

Using this equivalence and the law of asynchrony, we can deduce
$
\wait\lel x.P  ~\wbew~ \inp{\lel}x.(\out{\lel}x|\wait\lel x.P)
$.
\end{example}

An equivalence between \piL{} processes is also valid in \piw. To
state this property, given $P$ in \piL, we introduce \encw P, its
translation in \piw. The definition of \encw P is simple, as  we just
need to add wait constructs under restrictions for \encw P to be typable. 

\begin{lemma}\label{l:bisL:bisw}
  Suppose \typ{\G;\Rs}{P\bisL Q}. Then
      \typ{\G_w}{\encw P\bisw\encw Q} for some \piw{}
  typing environment $\G_w$.
\end{lemma}
This result shows that the addition of wait does not increase the
discriminating power of contexts.
We refer to Appendix~\ref{a:piw} for the definition of \encw P and a
discussion of the proof of Lemma~\ref{l:bisL:bisw}.

\section{Related and Future Work}\label{s:ccl}

The basic type discipline for lock names that imposes a safe usage of
locks by always
releasing a lock after acquiring it is discussed
in~\cite{DBLP:conf/unu/Kobayashi02}. This is specified using \emph{channel
  usages} (not to be confused with the usages of
Section~\ref{s:def:piw}). Channel usages  in~\cite{DBLP:conf/unu/Kobayashi02} are processes in a subset of CCS, and can
be defined in sophisticated ways to control the behaviour of
$\pi$-calculus processes.
The encoding of references in the asynchronous $\pi$-calculus studied
in~\cite{DBLP:conf/concur/HirschkoffPS20} is also close to how locks are
used in \piw. A reference is indeed a lock that must be released
\emph{immediately} after the acquire. The typed equivalence to reason about
reference names in~\cite{DBLP:conf/concur/HirschkoffPS20} has
important 
differences w.r.t.\  \wbew, notably because the deadlock- and
leak-freedom properties are not taken into consideration in that work.

\smallskip

The type system for \piw{} has several ideas in common with
\cite{DBLP:journals/pacmpl/JacobsB23}. 
That paper studies
\lamlo, a functional
language with higher-order locks and thread spawning.
The type system for \lamlo{} guarantees leak- and
deadlock-freedom by relying on 
duality and linearity properties, which 
entail the absence of cycles.
In turn, this approach originates in
work on binary session types,
 and in particular on concurrent versions of the
Curry-Howard
correspondence~\cite{DBLP:conf/esop/HondaVK98,DBLP:journals/jfp/GayV10,DBLP:journals/jfp/Wadler14,DBLP:conf/concur/CairesP10,DBLP:conf/esop/ToninhoCP13,DBLP:conf/esop/RochaC23}.

\piw{} allows a less controlled form of interaction than functional
languages or binary sessions. Important differences are: names do
not have to be used linearly; there is
no explicit notion of thread, neither a fork instruction, in \piw; reduction is not deterministic.
The type system for \piw{} controls parallel composition to rule out
cyclic structures among interacting processes.

The simplicity of the
typing rules, and of the proofs of deadlock- and leak-freedom, can be leveraged
to develop 
a theory of typed
behavioural equivalence for \piL{} and \piw. Soundness of bisimilarity
provides a useful tool to establish
equivalence results.  Proving completeness is not obvious, intuitively
because the constraints imposed by typing prevent us from adapting
standard approaches.
The way \bisw{} is defined should allow us to combine locks
with other programming constructs in order to reason about
programs featuring locks and, e.g., functions,
continuations, and references. Work in this direction will build 
   on~\cite{DBLP:journals/mscs/Milner92,DBLP:journals/iandc/Sangiorgi94,DBLP:conf/lics/DurierHS18,
  DBLP:conf/lics/HirschkoffPS21, DBLP:conf/icalp/Prebet22}.

\smallskip

 Our proofs of deadlock- and leak-freedom suggest that there is room
 for a finer analysis of how lock names are used. It is natural to try and extend our type system in order to accept more processes,
 while keeping the induced  behavioural equivalence tractable.
A possibility for this is to add \emph{lock
   groups}~\cite{DBLP:journals/pacmpl/JacobsB23}, with the aim of
 reaching an expressiveness comparable to the system
 in~\cite{DBLP:journals/pacmpl/JacobsB23}.
In a given lock group, locks are ordered, which makes it possible to
analyse systems having a cyclic topology.

Relying on orders to program with locks is a natural approach, that
has been used to define expressive type systems for lock freedom in
the
$\pi$-calculus~\cite{DBLP:conf/popl/IgarashiK01,DBLP:conf/unu/Kobayashi02,DBLP:conf/csl/Padovani14}.
In these works, some labelling is associated to channels or to actions
on channels, and the
typing rules
guarantee that it is always possible to define an order, yielding
lock-freedom.
We plan to study how our type system can be extended with lock groups
or ideas from type systems based on orders.

Rule~\rref{eq:wait} from Section~\ref{s:intro} explains in a concise
way how the wait operation behaves. Part of the difficulty in
Section~\ref{s:piw} is in defining a labelled semantics that is
compatible with the `magic' of executing a wait on \lel{} only when the
restriction can be put on top of the final release of \lel.
We plan to provide a more operational description of
deallocation, using, e.g., reference counting as
in~\cite{DBLP:journals/pacmpl/JacobsB23}. \piw{} could then be seen as
a language to describe at high-level what happens at a lower level when using and
deallocating locks.

\paragraph*{Acknowledgement.}
We are grateful to Jules Jacobs for an interesting discussion about
this work, and for suggesting the reduction rule~\rref{eq:wait} from
Section~\ref{s:intro}.
We also thank the anonymous referees for their helpful remarks and
advices.

\bibliography{refs-pilocks.bib}

\appendix
\section{Additional Material for Section~\ref{s:piL}}

\subsection{\CCSL, Operational Semantics}\label{a:CCSL}

Structural congruence is the least congruence satisfying the following
axioms:
\begin{mathpar}
  \inferrule{ }{P|Q\equiv Q|P}

  \inferrule{ }{P|(Q|R)\equiv (P|Q)|R}
  \\
  \inferrule{ }{P|(\new\lel) Q\equiv (\new\lel)(P|Q)}\text{ if
  }\lel\notin\flocks P

  \inferrule{ }{(\new\lel)(\new\lel')P\equiv(\new\lel')(\new\lel)P}
\end{mathpar}

To define reduction, we introduce \emph{execution contexts}, $E$,
given by $E~::=~\hole\OR E|P\OR (\new\lel)E$, where \hole{} is the
hole. $E[P]$ is the process obtained by replacing the hole in $E$ with
$P$.

Reduction is defined by the following rules:
\begin{mathpar}
  \inferrule{ }{\outC\lel|\lel.P \red P}

  \and
  
  \inferrule{ P\red P' }{ E[P]\red E[P'] }

  \and
  
  \inferrule{ Q\equiv P\and P\red P'\and P'\equiv Q' }{ Q\red Q' }
  \end{mathpar}

  \subsection{\CCSL, Properties of the Type System}

\begin{proof}[of Lemma~\ref{l:nocycle}]
  We show by induction on $k$ that $\ell_1.P_1|\dots|\ell_{k-1}.P_{k-1}$ is
  lock-connected: this holds because for every $i$, $\ell_i.P_i$ is
  lock-connected, and because 
  $\ell_i.P_i\sharel{\ell_i}\ell_{i+1}.P_{i+1}$ for all $i<k$.

  Moreover, we know $\ell_{k-1}.P_{k-1}\sharel{\ell_{k-1}}\ell_k.P_k$
  and  $\ell_{k}.P_{k}\sharel{\ell_k}\ell_1.P_1$. So names
  $\ell_{k-1}$ and $\ell_k$ belong to the free names both of
 $\ell_1.P_1|\dots|\ell_{k-1}.P_{k-1}$ and of $\ell_k.P_k$. By
 Lemma~\ref{l:lock:connect}, this prevents
 $\ell_1.P_1|\dots|\ell_{k}.P_{k}$
 from being typable.
\end{proof}

  \subsection{\piL, Operational Semantics}\label{a:piL}

  \newcommand{\eqr}{\ensuremath{\equiv_\mathrm{r}}}
  
Structural congruence in \piL, written $\equiv$, is standard, except
for the treatment of mismatch. Indeed, the corresponding axiom
cannot be used under an acquire prefix.

To handle this, we introduce an auxiliary structural congruence
relation, written \eqr.
Relation $\equiv$ is the smallest equivalence relation that satisfies
the axioms for $\equiv$ in \CCSL, plus the following ones
\begin{mathpar}
    \inferrule[PNil]{ }{P|\nil\equiv P}

  \inferrule[RNil]{ }{(\new\lel)\nil\equiv\nil}

  \inferrule[Mat]{ }{[v=v]P_1,P_2\equiv P_1}

    \inferrule[Mis]{ }{[v=v']P_1,P_2\,\equiv\, P_2}\text{ if } v\neq v'
\end{mathpar}
and also the contextual axioms 
\begin{mathpar}
  \inferrule[CPar]{P\equiv Q}{P|T\equiv Q|T}

  \inferrule[CRes]{P\equiv Q}{(\new\lel)P\equiv(\new\lel)Q}

  \inferrule[CAcq]{P\eqr
    Q}{\inp\lel{\lel'}.P\equiv\inp\lel{\lel'}.Q}

\end{mathpar}
The last axiom refers to \eqr, which is defined like $\equiv$, except that
\trans{Mis} is omitted and \trans{Cacq} is replaced by
\begin{mathpar}
  \inferrule[CAcq$_\mathrm{r}$]{P\eqr
    Q}{\inp\lel{\lel'}.P\eqr\inp\lel{\lel'}.Q}
\end{mathpar}

\paragraph*{Labelled Semantics for \piL.}

Actions of the LTS are defined by $\mu\quad::=\quad \inp\lel v\OR \out\lel
v\OR \bout\lel{\lel'}\OR \tau$.

The set of free names of $\mu$ is defined by $\flocks{\out\lel
  v}=\flocks{\inp\lel v}=\set{\lel,v}$, $\flocks\tau=\emps$ and
$\flocks{\bout\lel{\lel'}}=\set\lel$.

The set of bound names of $\mu$ is defined by $\blocks\mu = \emps$, except for
$\blocks{\bout\lel{\lel'}}=\set{\lel'}$.

The transition rules are the following:
\begin{mathpar}
  \inferrule{ }{\out\lel v\stra{\out\lel v}\nil}

  \inferrule{ }{\lel(\lel').P\stra{\inp\lel v}P\sub v{\lel'}}

  \inferrule{ P\stra{\out\lel{\lel'}}P' }{
    (\new\lel')P\stra{\bout\lel{\lel'}}P' }

  \inferrule{ P\stra\mu P' }{(\new\lel)P\stra\mu(\new\lel)P' }
  ~\lel\notin\flocks\mu

  \inferrule{ P\stra\mu P' }{ P|Q\stra\mu P'|Q}~\flocks
  Q\cap\blocks\mu=\emps

  \inferrule{ P\stra{\out\lel v}P'
    \and Q\stra{\inp\lel v}Q' }{
    P|Q\stra\tau P'|Q' }

  \inferrule{ P\stra{\bout\lel{\lel'}}P'
    \and Q\stra{\inp\lel{\lel'}}Q' }{
    P|Q\stra\tau (\new\lel')(P'|Q') }

  \inferrule{ P_1\stra\mu P'_1 }{[v=v]P_1,P_2\stra\mu P'_1 }

  \inferrule{ P_2\stra\mu P'_2 }{
    [v=v']P_1,P_2\stra\mu P'_2 }~v\neq v'
\end{mathpar}

\section{Additional Material from Section~\ref{s:piw}}\label{a:piw}

\subsection{Leak-Freedom in \piw}

The proof of Lemma~\ref{l:piw:prog} follows the approach of the
proof of Lemma~\ref{l:CCSL:prog}. 
An additional difficulty with respect to the latter proof is that release and wait obligations on a given
lock need not be explicit in the process, in the sense that they can be
stored in another lock.
      \begin{proof}
  We first consider the situation where $\typ\G P_0$, \G{} is
  complete, and we can write
  \begin{mathpar}
    P_0~\equiv~ \res{S}(
    \prod_i\out{\lel_i}{v_i}~|~ \prod_j\inp{\lel_j}{x_j}.P_j~|~
    \prod_k\wait{\lel_k}{y_k}.{Q_k})
    .
  \end{mathpar}
{We let $P = \prod_i\out{\lel_i}{v_i}~|~ \prod_j\inp{\lel_j}{x_j}.P_j~|~
\prod_k\wait{\lel_k}{y_k}.{Q_k}$.}
  
 We introduce some terminology to reason about this decomposition.   A \emph{prime process} is a process of the form $\out\lel v$,
 $\inp\lel x.P'$ or $\wait\lel y.P'$. Here ``prime'' refers to the fact
 that such processes cannot be decomposed modulo $\equiv$.
   We call \emph{subject} of a prime process the name that
  occurs in subject position in the topmost prefix of that process:
  these are the $\lel_i$s, the $\lel_j$s and the $\lel_k$s in the
  decomposition above. 

  We make the two following observations.
      First, for any $\lel\in\flocks P$, {either $\ell\in\flocks{P_0}$, or}
a release of \lel{} and a wait on
    \lel{} must be available, by typing.
                                            Second, none of the $\lel_i$ is equal to one of the $\lel_j$, since
    otherwise $P_0$ could reduce.     Moreover, if some $\lel_i$ is equal to one of the $\lel_k$s, then $P$
    necessarily contains an acquire on $\lel_i$, since otherwise $P_0$
    could reduce by performing a wait transition. {In the following,
    we do not consider the prime processes whose subject is in $\G$.
    Recall that these processes are outputs $\out{\ell_i}{v_i}$ with $v_i$
    being either of type $\bbool$ or $\tlockw{\ttyp}{00}$.}

    \medskip
    
      To derive a contradiction,   we show that the subject of every prime process
  occurs free in another prime process having a different subject.
    We examine the three forms of prime processes.
      \begin{itemize}
  \item 
  Consider first   $\inp{\lel_j}{x_j}.P_j$. The available \emph{release of $\lel_j$} cannot occur
  at top-level, since otherwise $P$ could reduce. The release cannot be
  available under an acquire or wait prefix on $\lel_j$, by typing and
  by definition of being available.

  The release of $\lel_j$ may be available in one of the $P_{j'}$s, or in one
  of the $P_k$s occurring under a prefix at some lock name different
  from $\lel_j$. In both cases, $\lel_j$ occurs in another prime
  process having a different subject.

 If the release on $\lel_j$ is available neither in the $P_j$s nor in
 the $P_k$s, then there exists another release
 of the form
 \out{\lel}{\lel_j} for some \lel, that does not occur under an
 acquire on $\lel_j$.
 We remark that \lel's usage is of the form $1w$, and that
 $\lel\neq\lel_j$. 
   
 Thus, the release of \lel{} necessarily occurs in a prime
 process whose subject is different from $\lel_j$.

\item Consider now $\wait{\lel_k}{y_k}.{P_k}$. As above, we reason about
  the \emph{release of $\lel_k$}. The only difference is that the
  release of $\lel_k$ may occur at top-level. If this is the case, then
  there is necessarily an acquire on $\lel_k$, otherwise $P$ could
  reduce. This acquire cannot occur at top-level, since otherwise $P$
  could reduce, by performing a wait transition. Hence, there is a prime
  process whose subject is different from $\lel_k$ that contains an
  acquire on $\lel_k$.

\item Consider \out{\lel_i}{v_i}. We reason about the \emph{wait on
    $\lel_i$}.
      If the wait on $\lel_i$ occurs at top-level, then, as above, an acquire on
  $\lel_i$ must occur in $P$, since otherwise $P_0$ could reduce.
  That acquire on $\lel_i$ cannot occur at top-level, since otherwise
  $P$ could reduce. So in this case $\lel_i$ occurs in a prime process
  whose subject is different from $\lel_i$.

  If the wait on $\lel_i$ does not occur at top-level, then it
  can occur in a prime process whose subject is different
  from $\lel_i$: that process cannot start with an acquire on 
  $\lel_i$ since otherwise $P$ could reduce.

  The last possibility is that
     a subterm  of the form \out\lel{\lel_i} occurs in some
 other prime process, and \lel{} carries the wait obligation.
 Reasoning as above,
  the subject of the prime process cannot be
  $\lel_i$.
  \end{itemize}
    We have shown that every prime process in the decomposition above
  whose subject is \lel{} can be connected
    with a different prime process.
        Like in the proofs of deadlock-freedom, we obtain a cycle, which is
  impossible by (the
 counterpart of) Lemma~\ref{l:nocycle}.

 \smallskip
 
\end{proof}

\subsection{Translating a \piL{} Process in \piw}

If $P$ is a \piL{} process, \encw P is its
translation into \piw, defined as follows:
\begin{mathpar}
  \encw{(\new\lel)P} = (\new\lel)(\encw P\,|\,\wait\lel x.\nil)

  \encw{\lel(\lel').P} = \lel(\lel').\encw P

  \encw{\out\lel v} = \out\lel v

  \encw{P_1|P_2} = \encw{P_1}\,|\,\encw{P_2}

  \encw{[v=v']P_1,P_2} = [v=v']\encw{P_1},\encw{P_2}
\end{mathpar}

To prove Lemma~\ref{l:bisL:bisw}, we establish a correspondence
between typing in \piL{} and in \piw.
If \typ{\G;\Rs}P, the typing
environment to type $P$ seen as a \piw{} process is constructed by
making sorts explicit, and by assigning usage $10$ for name \lel{} if
$\lel\in\Rs$, and $00$ otherwise.
Conversely, if $P\in\piL$ can be typed as a \piw{} process with
\typ{\G_w}P, then we can suppose that $\G_w$ does not contain any
usage of the form $r1$. We recover a \piL{} typing for $P$ by
collecting all names having type usage $10$ in \Rs, and erasing type
information in the components of $\G_w$, yielding \G, so that \typ{\G;\Rs}P.

This correspondence is extended to a correspondence between
transitions, so that a bisimulation relation in \piL{} is also a
bisimulation in \piw, via the aforementioned translation.
To prove the latter property, we rely on the equivalence
  \typ{\set{\set{v}}}{(\new\lel)(\out\lel v|\wait\lel x.\nil)
    \bisw\nil} in \piw.

\end{document}